 \documentclass[preprint]{elsarticle}
 \usepackage[a4paper, margin=1in]{geometry}
\pdfoutput=1

\usepackage{amssymb}
\usepackage{amsmath}
\usepackage{amsthm}
\usepackage{xcolor}
\usepackage{subcaption}

\usepackage{mathtools}
    
\newcommand{\OA}{\mathcal{O}}

\newcommand{\U}{\mathcal{U}}
\newcommand{\T}{\mathcal{T}}
\newcommand{\R}{\mathbb{R}}
\newcommand{\F}{\mathcal{F}}
\newcommand{\Y}{\mathcal{Y}}
\newcommand{\PA}{\mathcal{P}}
\newcommand{\SA}{\mathcal{S}}
\newcommand{\ZA}{\mathcal{Z}}
\newcommand{\sa}{\mathfrak{s}}
\newcommand{\ua}{\mathfrak{u}}

\newcommand{\Cb}{\mathbb{C}}
\newcommand{\Cf}{\mathfrak{C}}
\newcommand{\Bb}{\mathbb{B}}
\newcommand{\Bf}{\mathfrak{B}}

\newcommand{\RR}{\mathbb{R}} 
\newcommand{\X}{\mathcal{X}}

\newtheorem{theorem}{Theorem}
\newtheorem{lemma}[theorem]{Lemma}
\newtheorem{proposition}[theorem]{Proposition}
\newtheorem{remark}{Remark}

\newcommand{\comm}[1]{}

\numberwithin{equation}{section}

\newcommand{\ignore}[1]{ }

\begin{document}

\begin{frontmatter}



\title{Reinforcement Learning in Non-Markovian Environments}

\author[SC]{Siddharth Chandak}\corref{cor1}
\address[SC]{Department of Electrical Engineering, Stanford University, Stanford, CA 94305, USA}
\ead{chandaks@stanford.edu}

\author[PS]{Pratik Shah}
\ead{pratik2002shah@gmail.com}
\address[PS]{Department of Mechanical Engineering,  Indian Institute of Technology Bombay,  Powai, Mumbai-400076, India} 

\author[VSB]{Vivek S Borkar\corref{cor1}}
\ead{borkar.vs@gmail.com}
\address[VSB]{Department of Electrical Engineering,  Indian Institute of Technology Bombay,  Powai, Mumbai-400076, India} 

\author[SC]{Parth Dodhia}
\ead{pdodhia@stanford.edu}

\cortext[cor1]{Corresponding author}

\begin{abstract}
Motivated by the novel paradigm developed by Van Roy and coauthors for reinforcement learning in arbitrary non-Markovian environments, we propose a related formulation and explicitly pin down the error caused by non-Markovianity of observations when the Q-learning algorithm is applied to this formulation. Based on this observation, we propose that the criterion for agent design should be to seek good approximations for certain conditional laws. Inspired by classical stochastic control, we show that our problem reduces to that of recursive computation of approximate sufficient statistics. This leads to an autoencoder-based scheme for agent design which is then numerically tested on partially observed reinforcement learning environments.
\end{abstract}





\end{frontmatter}


\section{Introduction}
In a series of remarkable works, Van Roy and coauthors laid down the theoretical foundations of agent design for reinforcement learning in an arbitrary environment \citep{Dong,Lu}.
The model in \cite{Dong} postulates an observed process $\{O_n\}$ taking values in a space of observations $\OA$, an agent state $\{S_n\}$ taking values in a state space $\SA$, and a control sequence $\{U_n\}$ taking values in an action space $\U$. The agent state is given recursively by
\begin{equation}
S_{n+1} = f(S_n,U_n, O_{n+1},\xi_{n+1}), \ n \geq 0, \label{state-eq}
\end{equation}
for some $f : \SA\times\U\times\OA\times\X \mapsto \SA$. Here $\{\xi_n\}$ is an independent and identically distributed  sequence of random variables taking values in some Polish space $\X$, representing extraneous noise. As a simplification, we henceforth drop $\xi_n$ from $f(\cdot)$ in \eqref{state-eq}, noting that this noise can be absorbed into the controls as randomized controls.  
At time $n$, a reward $r(S_n,U_n,O_{n+1})$ is received for a prescribed $r : \SA\times\U\times\OA \mapsto [0,1]$. The main result of \cite{Dong} is to establish a regret bound for the reward.
The model is further refined in 
 \cite{Lu}, but we stick to the simpler paradigm of \cite{Dong} above.
 
 {\color{black} The key point to note here is that the observables $\{O_n\}$  from the environment are \textbf{not} assumed to be Markov. This is the reality in most situations, the Markov model we impose is an approximation. This is explicitly so  e.g., when the model used is a discrete or finite dimensional caricature of a more complex problem, or because the convenient approximations imposed on the dynamics for analytic ease are only approximate (such as exponentiality of interarrival times in controlled queues). Also note that the agent dynamics (\ref{state-eq}) above, inclusive of the choice of the agent state, is a \textbf{design choice} we impose when we postulate the model. In \cite{Dong, Lu}, it is explicitly identified as such. Usually the physics of the problem may dictate a natural choice, but when that is not the case, a principled approach is called for. This design issue is the main problem we plan to address, which we come to later in this work after preparing the theoretical background for it. The theory, however, is built taking the model (\ref{state-eq}) as a given.}

In this work, we take a somewhat different view of this problem. We take the observation process $O_n$ to have been evolving forever, i.e., for $n > - \infty$, which allows us some simplification.  Let $O^n$ denote the semi-infinite sequence $[O_n, O_{n-1}, \cdots] \in \OA^\infty$ and let $U^n$ denote the sequence $[U_n, U_{n-1}, \cdots] \in \U^\infty$. Analogous notation will be used for other random processes. 
$\PA(\cdots)$ will always denote the space of probability measures on the Polish space `$\cdots$' with Prohorov topology, also known as the `topology of weak convergence' \citep{KRP}. For simplicity, we take $\SA, \U, \OA$ to be finite sets. {\color{black} By Tychonoff's theorem, $U^n, \SA^n, \OA^n$ are compact spaces under the product topology for $1\leq n \leq \infty$.}

Our main contributions are the following.
\begin{enumerate}
\item  By explicitly pinning down  the error caused by potential non-Markovianity of the observation process in the application of the Q-learning algorithm \citep{Watkins}, we propose an alternative performance criterion for agent design. 

\item We argue that this reduces the problem to seeking good approximations for certain conditional laws.

\item We justify achieving such approximation by `recursively computable approximate sufficient statistics' (RCASS), drawing upon the classical notion of Partially Observed Markov Decision Processes (POMDPs) in stochastic control.

\end{enumerate}

The notion of RCASS is neither very novel nor surprising. In fact, it is essentially the same as the notion of an `approximate information state' proposed and analyzed in the excellent recent article \cite{Aditya}, which also includes a perspective on how such notions have evolved in related literature, both classical and contemporary (see section 2.5 and the appendices of \textit{ibid.}). The difference is in the way we motivate it. We take Q-learning for discounted cost as a test case, a similar treatment should be possible in principle for other learning schemes as well. We show that if we use a surrogate state model as in (\ref{state-eq}) and treat it  as Markovian (though it is not) for purposes of the Q-learning algorithm, the latter still converges a.s.\ to a certain limit that can be interpreted. Furthermore, we give an analysis of transient errors where the error due to non-Markovianity  can be isolated clearly and should in fact be the basis for choosing the agent state dynamics (\ref{state-eq}), viz., with the objective of approximating certain conditional laws. This naturally leads to RCASS. Of course, we could target approximation of the relevant conditional \textit{expectations} alone, in this case of the value function estimate and the reward, which would be closer to the viewpoint taken in \cite{Aditya}. We opt for approximation of conditional laws because of its kinship with the classical `separated control problem' for POMDPs. This is the problem of controlling the completely observed Markov decision process of conditional laws of the state given past observations and controls (aka the `belief states'), given by the recursive nonlinear filter. This is dubbed `separated control' because the estimation and control aspects get separately delineated. Thus we view the dynamics (\ref{state-eq}) as a finite state approximation to the nonlinear filter. This point will be discussed later in this work along with   proposals for schemes to learn the agent dynamics. At a meta level, our attempt is to make a connection between reinforcement learning in general environments and classical stochastic control \citep{Kumar}. In the process, we also highlight a parallel with the classical `internal model principle' of control theory that deserves to be better appreciated as a guiding principle.

In Section 2, we begin by recalling the Q-learning algorithm applied to the agent state described above. We show that this algorithm converges to a fixed point which can be used to approximate the true $Q$-values corresponding to the policy based on the entire history and make control decisions accordingly. We also present a finite-time analysis, where using a suitable `decomposition' of the $Q$-learning iterate, we isolate different error terms along the lines of \cite{Chandak}, where all but the new term that captures the error due to non-Markovianity are inherited from \cite{Chandak} and, $(i)$ are applicable to the controlled Markov case as well, and, $(ii)$ are asymptotically negligible. The error due to non-Markovianity also vanishes asymptotically, but the concentration bounds, when available, are weaker, suggesting non-negligible finite time errors. We dub this difficulty in approximating the true $Q$-values and this additional error term as the `\textit{the curse of non-Markovianity}'. Hence a reasonable target is to  minimize them `to the extent possible'. This is precisely what gives us the criterion of well approximating conditional laws by a judicious choice of `approximate sufficient statistics' or, to get closer to machine learning parlance, `dynamic features'.

Section 3 uses the foregoing to make a case for approximation by a simpler model using recursively computable approximate sufficient statistics, with some rigorous justification. This amounts to developing an approximate dynamic model for the controller. This idea is not new, see, e.g., \cite{Yu}, which seeks finite state machines as candidate controllers for POMDPs and establishes their near-optimality. In the latter context, these finite state controllers can be thought of as approximate controlled nonlinear filters, or approximate POMDPs. We point out these parallels.

{\color{black}In Section 4, we present a computation scheme inspired by the autoencoder model from \cite{Champion}. Unlike autoencoders, where one tries to match the input with the output, our architecture tries to match the output with the next input (i.e., tries to predict the next observation). The agent state in this architecture is given by the output of the encoder. These agents states are passed to a Deep Q-Network (DQN) \cite{dqn} which then learns the control policy. We call the agent trained in this manner as the \textit{Non-Markovian Q-Agent}. This agent is then numerically tested on three environments. In the concluding section on future directions, we sketch another plausible scheme for learning the agent dynamics and point out issues with the stationarity assumption. }

\section{The curse of non-Markovianity}

\subsection{The model and the algorithm}

Recall the notation from the previous section. Our agent dynamics is given by
\begin{equation}
S_{n+1} = f(S_n, U_n, O_{n+1}), \ n > -\infty. \label{agentdyn}
\end{equation}
Here $\{O_n\}$ is the observed process, taking values in a finite space of observations $\OA$, $\{S_n\}$ is the agent state taking values in a finite state space $\SA$ with $|\SA|=\sa$,  $\{U_n\}$ is the control sequence  taking values in a finite action space $\U$ with $|\U|=\ua$.
We assume that the observation process $O_n$ has been evolving forever. We also assume that $\{(U_n,O_n)\}$ is stationary ergodic\footnote{\color{black} We argue later that without loss of generality, we may take $\{(S_n,U_n,O_n)\}$ to be stationary ergodic, see Remark \ref{R1} below.} (more generally, asymptotically stationary and ergodic will suffice) and satisfies 
\begin{equation}
P(O_{n+1} \in A\mid  O^n,U^n) =p(A\mid O^n,U^n), \label{kernel}
\end{equation}
for a continuous map 
$$(o,u) \in \OA^\infty\times\U^\infty \mapsto p(\cdot\mid o,u) \in \PA(\OA).$$ 
We take $\{U_n\}$ to be a stationary randomized policy, i.e., $$P(U_n=u\mid S_n=s)=\phi(u\mid s),$$ for a map $s \in \SA \mapsto \phi(\cdot\mid s) \in \PA(\U).$ {\color{black} We can construct measures on one sided infinite product spaces $\prod_{m=n}^\infty(\OA_m\times\U_m)$ for $n > -\infty$  by setting $(O_k,U_k) \equiv$ some fixed $(o,u) \in \OA\times\U$ for $k < n$. By the Ionescu-Tulcea theorem, the two maps defined above then uniquely define a probability measure $\widetilde{P}_n$ on 
$\prod_{m=n}^\infty(\OA_m\times\U_m)$, where $\OA_m,\U_m$ are replicas of $\OA,\U$. What we have here may be viewed as its inductive limit as $n\to-\infty$. }

We also assume that $P(S_n=s,U_n=u)$ is bounded away from zero for all $s\in\SA,u\in\U$ under the stationary distribution. Let $\tilde{\pi}(s,u)$ for $(s,u)  \in \SA^{\sa\times\ua}$ denote this stationary law for $(S_n,U_n)$. Similarly, let $\pi(\cdot)$ denote the stationary distribution for $(O^n,U^n)$. 

Let $r(S_n,U_n,O_{n+1})$ be the reward at time $n$, where $r : \SA\times\U\times\OA \mapsto [0,1]$. For $s\in\SA$ and $u\in\U$, define
$$\bar{r}(s,u) := E[r(S_n,U_n,O_{n+1})\mid S_n = s, U_n = u].$$ 
Consider the Q-learning algorithm of \cite{Watkins} applied to the agent dynamics above, disregarding the non-Markovianity of $\{S_n\}$. It leads to:
\begin{align}
Q_{n+1}(s,u)& = Q_n(s,u) + a(n)I\{S_n = s, U_n = u\}
\Big[r(S_n,U_n,O_{n+1})   \nonumber\\
&\;\;\;\;\;\;\;\;\;\;\;\;\; + \gamma\max_aQ_n(S_{n+1}, a)- Q_n(s,u)\Big]. \label{Qlearn}
\end{align}
Here $I\{\cdots\}$ is the `indicator' function that is $1$ if `$\cdots$' holds and $0$ otherwise. Also, $a(n)$ is a sequence of non-negative stepsizes satisfying the Robbins-Monro conditions
$$\sum_n a(n)=\infty, \; \sum_n a(n)^2<\infty.$$
For simplicity, we also assume that $Q_0(s,u)\in[0,\frac{1}{1-\gamma}]$ for all $s,u$. This inductively implies that $Q_n(s,u)\in[0,\frac{1}{1-\gamma}]$ for all $s,u$ and $n>0$\footnote{This uniform boundedness remains valid without this assumption, but with a different bound involving also the initial condition.}.  Let $q(s'\mid s,u) := P(S_{n+1} = s'\mid S_n = s, U_n  = u)$, where the absence of explicit time dependence of $q$ is due to our stationarity hypothesis. Then
\begin{align}
    Q_{n+1}(s,u)&=Q_n(s,u)+a(n)\Big(F^{s,u}(Q_n,S_n,U_n)\nonumber\\ &\;\;\;\;\;\;\;\;\;\;\;\;\;\;\;\;\;\;\;\;+\zeta^{s,u}(Q_n,O^n,U^n)+M^{s,u}_{n+1}\Big). \label{Qlearn2}
\end{align}
Here,
\begin{itemize}
\item  $F(Q_n,S_n,U_n)$ is the usual term observed in Q-learning for MDPs, with its $(s,u)$\textsuperscript{th} element given by
\begin{align*}
F^{s,u}(Q_n,S_n,U_n)&\coloneqq I\{S_n = s, U_n = u\}\Big[\bar{r}(s,u)\\
+& \gamma\sum_{s'}q(s'\mid s,u)\max_aQ_n(s', a) - Q_n(s,u)\Big],
\end{align*}

\item $\zeta(Q_n,O^n,U^n)$ is the offset due to non-Markovianity, with its $(s,u)$\textsuperscript{th} element defined as
\begin{align*}
&\zeta^{s,u}(Q_n,O^n,U^n) \\
&\coloneqq I\{S_n=s,U_n=u\}\Bigg[E[r(S_n,U_n,O_{n+1})\mid O^n,U^n]\nonumber \\
&\;\;-\bar{r}(S_n,U_n)+\gamma\bigg(\sum_{s'}P(S_{n+1} = s'\mid O^n,U^n)\max_aQ_n(s',a) \\ &\;\;\;\;\;\;\;\;\;\;\;\;-\sum_{s'}q(s'\mid S_n,U_n)\max_aQ_n(s', a)\bigg)\Bigg],  
\end{align*}
and,

\item  $M_{n+1}\in\RR^{\sa\ua}$ is a martingale difference sequence with its $(s,u)$\textsuperscript{th} element defined as
\begin{align*}
&M^{s,u}_{n+1}\coloneqq I\{S_n=s,U_n=u\}\Bigg[\bigg(r(S_n,U_n,O_{n+1}) \\
&\;\;\;\;\;\;- E[r(S_n,U_n,O_{n+1})\mid O^n,U^n]\bigg)+\gamma\bigg(\max_aQ_n(S_{n+1}, a)\\ &\;\;\;\;\;\;\;\;\;\;\;\;-\sum_{s'}\max_aQ_n(s',a)P\left(S_{n+1} = s'\mid O^n,U^n\right)\bigg)\Bigg].
\end{align*}
\end{itemize} 

{\color{black} The following are a few remarks about the map $p(\cdot\mid\cdot,\cdot)$, the state process $\{S_n\}$ and the dependence of the rewards on it.

\begin{remark}\label{R1}
    The continuity condition on the map $p( \cdot\mid\cdot,\cdot)$ is not probabilistic, i.e., not `a.s.', but a deterministic constraint on this map and is independent of our statistical hypotheses on the processes. There is, however, one important observation worth note here. Ergodicity implies that the tail $\sigma$-field at $-\infty$ given by $\mathcal{F}_{-\infty} := \cap_{n> >-\infty}\sigma(O^n,U^n)$ is trivial, i.e., the completion of $\{\phi,\Omega\}$ where $(\Omega,\F,P)$ is the underlying probability space. That is, there is no information from the infinite past. Thus in view of \eqref{agentdyn}, we may assume that $S_n = g(O^n,U^{n-1}) \ \forall n$ for some measurable $g: \OA^\infty\times\U^\infty \mapsto \SA$. Intuitively, this says that $\{S_n\}$ was initiated in infinite past with a deterministic `initialization'. The lack of a time subscript $n$ for $g$, however, is due to stationarity. It may be observed that in what follows, all conditional expectations given $(O^n,U^n)$ are functions of one or more of $S_n,U_n,O_{n+1}$ and it is only the last one  of this triplet that is being averaged, the other two being measurable w.r.t.\ $(O^n,U^n)$. For $O_{n+1}$, the conditional probability or expectation is with respect to the continuous transition kernel $p(\cdot|\cdot,\cdot)$. Typical scenarios when one may have such a condition is a well-structured continuous dependence on the past such as dependence on a finite window in the past with deterministic length or i.i.d.\ random length, or an exponentially weighted average. Partially observed Markov decision processes that are aperiodic and irreducible under stationary policies is a rich source of examples. (Periodicity for an irreducible Markov chain leads to a non-trivial $\mathcal{F}_{-\infty}$ that carries information about the `phase' of the realization of the sequence of periodic classes.)
\end{remark}
\begin{remark}
    While the reward a priori may not depend on the agent state $S_n$ which is an artificial construct, we are allowing such a dependence in case one wants to factor in `costs' such as computational or hardware overheads from the actual implementation of the agent as negative rewards into the overall reward.
\end{remark}

In the following subsection, we first show that the iterates $Q_n$ converge and then study how the limiting $Q^*$ gives us a criterion for agent design.}

\subsection{Asymptotic analysis}
Our first result states the convergence of $\{Q_n\}$.
\begin{theorem}\label{thm:convergence}
The iterates $\{Q_n\}$ from (\ref{Qlearn}) converge a.s.\ to $Q^*$ where $Q^*$ is the unique solution to the system of equations
\begin{equation}
Q(s,u)=\bar{r}(s,u)+\gamma\sum_{s'} q(s'\mid s,u)\max_aQ(s',a). \label{limit}
\end{equation}

\end{theorem}
\begin{proof}
Let $Y_m\coloneqq\{O^m,U^m\}$ for $m>0$ where $Y_m\in\Y\coloneqq\OA^\infty\times\U^\infty$. Then $\{Y_m\}$ is a Markov chain. From the definition of $S_m$ and the ergodicity assumption which in particular implies triviality of the tail $\sigma$-field of $\{Y_n\}$, we can conclude that $\{S_m,U_m\} \coloneqq \ell(O^m,U^m)$ for some measurable function $\ell$. Rewriting eq. (\ref{Qlearn2}) in vector form, we get
$$Q_{n+1}=Q_n+a(n)(F(Q_n,S_n,U_n)+\zeta(Q_n,O^n,U^n)+M_{n+1}).$$
 Also, define $h(Q_n,Y_n)\coloneqq F(Q_n,S_n,U_n)+\zeta(Q_n,O^n,U^n)$.

Now, using Corollary 8.1 from \cite{BorkarBook}, applicable here because of the compact state space, we conclude that $\{Q_n\}$ almost surely tracks the asymptotic behavior of the ODE
$$\dot{Q}(t)=\int_\Y h(Q(t),y)\pi(dy),$$
in the sense that it converges to an internally chain transitive invariant set of this ODE. (See \textit{ibid.} for a definition.) Here, recall that  $\pi(\cdot)$ is the stationary probability for the chain $\{Y_n\}$, i.e., the stationary law of $Y_n = (O^n, U^n)$.
Note that the stationary expectation of $\zeta(Q,O^n,U^n)$ is zero (i.e., $\int_\Y \pi(dy)\zeta(Q,y)=0$) for all $Q$. This implies that $\{Q_n\}$ tracks the ODE $\dot{Q}(t)=\sum_{s,u}F(Q(t),s,u)\tilde{\pi}(s,u)$. We have already assumed that $\tilde{\pi}(\cdot,\cdot)$ is bounded away from zero. Then using Corollary 12.1 from \cite{BorkarBook}, it follows that this ODE converges to $Q^*$, where $Q^*$ is the unique solution to the set of equations
$$Q(s,u)=\bar{r}(s,u)+\gamma\sum_{s'} q(s'\mid s,u)\max_aQ(s',a).$$
By standard arguments (\cite{BorkarBook}, Chapter 2, see also Chapter 12), $Q_n \to Q^*$ a.s.
This completes the proof for Theorem \ref{thm:convergence}.
\end{proof}

This result tells us that even when the controlled Markov model is an approximation, the algorithm does converge a.s.
 It also identifies the limit. We next analyse this limit in order to make it a basis for agent design. By virtue of the ergodicity assumption, we have $S_n = g(O^n,U^{n-1})$ for some $g: O^\infty\times\U^\infty \mapsto\SA$. {\color{black} Let $X_n := (O^n, U^{n-1})$. Then $\{X_n\}$ is a controlled Markov chain wherein the transition of from $X_n$ to $X_{n+1} = ((O^n, \tilde{o}), (U^{n-1}, a))$ under control $U_n = a$ is according to
\begin{eqnarray*}
\lefteqn{P(X_{n+1} = ((O^n, \tilde{o}), (U^{n-1}, a))| X_{n}, U_n = a)} \\
 &=& p(\tilde{o}| (O^n, (U^{n-1},a)).
\end{eqnarray*} 
Since we are using a fixed randomized policy 
$$\phi(U_n = \cdot \ | S_n) = \tilde{\phi}(U_n = \cdot \ | g(X_n))$$ 
for a suitably defined $\tilde{\phi}$, $\{X_n\}$ is in fact an $\OA^\infty\times\U^\infty$-valued Markov chain. Thus the sets $\{S_n = i\}$ define a partition of the state space of $\{X_n\}$, allowing us to view agent design as a state aggregation problem.}
 Hence the formalism of \cite{Singh} (see also \cite{Ben}  for related work) can be expected to apply, albeit now for an infinite state space. The exact counterpart of Theorem 1 therein is stated below in our notation. 
 \begin{proposition} $\{Q_n\}$ converge a.s.\ to the unique solution of the system of equations
\begin{align}\label{Singh1}
&Q(s,u) = \int_{\OA^\infty\times\U^\infty}P(dz\mid S_0=s,U_0=u) \ \times \nonumber \\
&\Bigg(\int_{\tilde{o}\in\OA, s'\in \SA}P(s', \tilde{o}\mid (O^0,U^0) = z)\Big(r(s,u,\tilde{o}) + \gamma\max_aQ(s',a)\Big)\Bigg).
\end{align}
\end{proposition}
This reduces precisely to the statement of Theorem \ref{thm:convergence}, as can be easily verified. Note, however, that our scheme is more than just a partitioning of $\OA^\infty\times\U^\infty$ : it is a partition that is preserved under the dynamics. Thus the agent design should aim towards minimizing the error due to such aggregation. 

We now give a result which explains the relevance of the fixed point $Q^*$ learnt using the iteration \eqref{Qlearn} and gives further motivation for the agent design. Let $\widehat{Q}$ \ be the unique solution in $C(\OA^\infty\times\U^\infty)$ (guaranteed by the standard contraction mapping argument) to the following system of equations: {\color{black}
\begin{align*}
&Q(x_n,a)=Q((o^n,u^{n-1}), a) = \\
&E[r(S_n,U_n,O_{n+1})|O^n = o^n,U^{n-1} = u^{n-1}, U_n = a] \\
    &+\gamma \int P(d\tilde{o}|(O^n,U^{n-1},a)\max_b Q(((O^n,\tilde{o}),U^{n}),b).
\end{align*}
This $\widehat{Q}$ represents the $Q$-values corresponding to the optimal policy for $\{X_n\}$ based on the complete history. } The following theorem shows that if the conditional distribution of $O_{n+1}$ given $(O^n,U^n)$ is well approximated by its conditional distribution given $(S_n,U_n)$, then $\widehat{Q}$ is well approximated by $Q^*$. 

\begin{theorem}\label{thm:approx}
    For any observation and action sequence, $O^n$ and $U^n$, $n>0$, let $d_{TV}(P( \cdot |S_n,U_n), p( \cdot |O^n,U^n))$ denote the total variation distance between the conditional distribution of $O_{n+1}$ given $(S_n,U_n)$ and $(O^n,U^n)$, respectively. If $$d_{TV}(P( \cdot |S_n,U_n), p( \cdot |O^n,U^n))\leq \varepsilon,$$ for all $n>0, O^n\in\mathcal{O}^{\infty},U^n\in\mathcal{U}^{\infty},$ then 
    $$\sup_m\sup_{O^m,U^m} \left|Q^*(S_m,U_m)-\widehat{Q}(X_m,U_m)\right|\leq \frac{\varepsilon}{(1-\gamma)^2}.$$
\end{theorem}
\begin{proof}
    For any $n>0$, observation sequence $O^n$, control sequence $U^{n}$, and corresponding state $S_n$, 
    \begin{align*}
        &\left|Q^*(S_n,U_n)-\widehat{Q}(X_n,U_n)\right|\\
        &\leq \big|\bar{r}(S_n,U_n)-E[r(S_n,U_n,O_{n+1})|O^n,U^n]\big|\\
        &\;\;+\gamma\Bigg|\sum_{s'} q(s'\mid S_n,U_n)\max_aQ^*(s',a)\\
        &\;\;\;\;\;\;\;\;-\int p(d\tilde{o}|O^n,U^n)\max_a \widehat{Q}((X_n,(\tilde{o},U_n)),a)\Bigg|.
    \end{align*}
    Note that $(X_n,(\tilde{o},U_n))$ here denotes $X_{n+1}$ for $O_{n+1}=\tilde{o}$. For the first term above, note that
    \begin{align*}
        &\big|\bar{r}(S_n,U_n)-E[r(S_n,U_n,O_{n+1})|O^n,U^n]\big|\\
        &= \big|E[r(S_n,U_n,O_{n+1})|S_n,U_n]-E[r(S_n,U_n,O_{n+1})|O^n,U^n]\big|\\
        &\leq d_{TV}(P( \cdot |S_n,U_n), p( \cdot |O^n,U^n)).
    \end{align*}
    The inequality follows from our assumption that the rewards are bounded in $[0,1]$. Simplify the second term as follows:
    \begin{align*}
        &\Bigg|\sum_{s'} q(s'\mid S_n,U_n)\max_aQ^*(s',a)\\
        &\;\;\;\;\;\;\;\;-\int p(d\tilde{o}|O^n,U^n)\max_a \widehat{Q}((X_n,(\tilde{o},U_n)),a)\Bigg|\\
        &=\Bigg|\int p(d\tilde{o}|S_n,U_n)\max_aQ^*\left(f(S_n,U_n,\tilde{o}),a\right)\\
        &\;\;\;\;\;\;\;\;-\int p(d\tilde{o}|O^n,U^n)\max_a \widehat{Q}((X_n,(\tilde{o},U_n)),a)\Bigg|\\
        &\leq\Bigg|\int p(d\tilde{o}|S_n,U_n)\max_aQ^*\left(f(S_n,U_n,\tilde{o}),a\right)\\
        &\;\;\;\;\;\;\;\;-\int p(d\tilde{o}|O^n,U^n)\max_a Q^*\left(f(S_n,U_n,\tilde{o}),a\right)\Bigg|\\
        &\;\;+\Bigg|\int p(d\tilde{o}|O^n,U^n)\max_aQ^*\left(f(S_n,U_n,\tilde{o}),a\right)\\
        &\;\;\;\;\;\;\;\;-\int p(d\tilde{o}|O^n,U^n)\max_a \widehat{Q}((X_n,(\tilde{o},U_n)),a)\Bigg|\\
        &\stackrel{(a)}{\leq} \frac{1}{1-\gamma} d_{TV}(P( \cdot |S_n,U_n), P( \cdot |O^n,U^n))\\
        &\;\;+\sup_m\sup_{O^m,U^m} \left|Q^*(S_m,U_m)-\widehat{Q}(X_m,U_m)\right|.
    \end{align*}
    The first term in inequality $(a)$ follows from the fact that $Q^*$ and $\hat{Q}$ always lies in $[0,1/(1-\gamma)]$. 
    
    This implies that for all $n$ and sequences $O^n,U^n$,
    \begin{align*}
         &\left|Q^*(S_n,U_n)-\widehat{Q}(X_n,U_n)\right|\\
         &\leq \varepsilon+ \gamma\left(\frac{\varepsilon}{1-\gamma}+\sup_m\sup_{O^m,U^m} \left|Q^*(S_m,U_m)-\widehat{Q}(X_m,U_m)\right|\right)
    \end{align*}
    Finally,
    \begin{align*}
        &\sup_m\sup_{O^m,U^m} \left|Q^*(S_m,U_m)-\widehat{Q}(X_m,U_m)\right|\\
        &\leq\frac{\varepsilon}{1-\gamma}+ \gamma \sup_m\sup_{O^m,U^m} \left|Q^*(S_m,U_m)-\widehat{Q}(X_m,U_m)\right| \\
        &\leq \frac{\varepsilon}{(1-\gamma)^2},
    \end{align*}
    which completes the proof for Theorem \ref{thm:approx}.
\end{proof}

It is clear from this analysis that the criterion for agent design should be to well approximate certain conditional expectations given the entire history till $n$, by the corresponding conditional expectations given $S_n,U_n,$ alone. Since we do not know the objects under the conditional expectation operator a priori, it makes sense to approximate the conditional laws themselves, viz., the conditional laws given the entire past by those given $S_n,U_n,$  alone. This issue is further highlighted by the concentration phenomena we study next.

\subsection{Finite-time Analysis}
We now present a concentration inequality for the iterates in (\ref{Qlearn}) and isolate the error due to non-Markovianity. Let $Z_n\coloneqq (S_n,U_n)$, where $Z_n\in\ZA\coloneqq \SA\times\U$. Then \begin{eqnarray*}
P(Z_{n+1}=(s',u')\mid Z_n=(s,u))&=&\psi((s',u')\mid (s,u))\\
&:=& q(s'\mid s,u)\phi(u'\mid s').
\end{eqnarray*}
As defined before, $\tilde{\pi}(z)$ is the stationary distribution for $z=(s,u)$. An important observation then is that $\tilde{\pi}$ is necessarily the stationary distribution for the transition probability kernel $\psi( \cdot \mid  \cdot)$. This allows us to define $V(\cdot,\cdot):[0,\frac{1}{1-\gamma}]^{\sa\ua}\times\ZA\mapsto\R^{\sa\ua}$ to be the solution of the Poisson equation:
$$V(q,z)=F(q,z)-\sum_{z'\in\ZA} \tilde{\pi}(z')F(q,z')+\sum_{z'\in\ZA}\psi(z'\mid z)V(q,z'),$$
for each $z\in\ZA$. The Poisson equation specifies $V(q,\cdot)$ uniquely only up to an additive constant. Fixing $V(q,z_0)=0$ for a fixed $z_0\in\ZA$ and for all $q$ gives us a unique solution to the set of Poisson equations given above. Henceforth $V$ refers to this unique solution of the Poisson equation. Then define 
\begin{align*}
    \omega(Q_n,O^n,U^n)&\coloneqq E[V(Q_n,Z_{n+1})\mid O^n,U^n]\\
    &\;\;\;\;\;\;\;\;\;\;-E[V(Q_n,Z_{n+1})\mid S_n,U_n].
\end{align*}

Similar to \cite{Chandak}, we make some additional assumptions about the stepsizes $a(n)$ required for the concentration bound. We assume that there exists $N>0$ such that for all $n\geq N$, $a(n+1)\leq a(n), a(n)<1$ and $\frac{d_1}{n}\leq a(n) \leq d_3 \left(\frac{1}{n}\right)^{d_2}$, where $d_1, d_3>0$ and $0.5<d_2\leq1$. 

Define $\pi_{min}=\min_{s,u}\tilde{\pi}(s,u)$. For $n\geq0$, define:
\begin{eqnarray*}
    b_k(n) &=& \sum_{m=k}^na(m), \ 0\leq k\leq n<\infty, \\
   \beta_k(n) &=& \begin{cases}
      \frac{1}{k^{d_2-d_1}n^{d_1}}, & \text{if}\ d_1\leq d_2 \\
      \frac{1}{n^{d_2}}, & \text{otherwise},
    \end{cases} \\
    \chi(n,m)&=&\begin{cases}
    \prod_{k=m}^n(1-a(k)), & \text{if}\ n\geq m\\
    1, &\text{otherwise},
    \end{cases}\\
    \Delta_k(n)&=&\sum_{m=k}^{n-1}\chi(n-1,m+1)a(m)\Big(\zeta(Q_m,O^m,U^m)\\
    &&\;\;\;\;\;\;\;\;\;\;\;\;\;\;\;\;\;\;\;\;\;\;+ \ \omega(Q_m,O^m,U^m)\Big), \ n > k.
\end{eqnarray*}

We then have:
\begin{theorem}\label{thm:conc} Let $n_0\geq N$. Then there exist finite positive constants $c_1$, $c_2$, and $D$, depending on $\|Q_{N}\|$, such that for $\delta_1>0$ and $n\geq n_0$, the inequality
\begin{equation}
    \|Q_n-Q^*\|_\infty \leq e^{-(1-\alpha)b_{n_0}(n)}\|Q_{n_0}-Q^*\|_\infty + \ \frac{\delta_1+a(n_0)c_1}{1-\alpha} + \|\Delta_{n_0}(n)\|_\infty, \label{Qbound}
\end{equation}
holds with probability exceeding 
\begin{align}
&1 -  2\sa\ua\sum_{m=n_0+1}^ne^{-D\delta_1^2/\beta_{n_0}(m)} , \ 0 < \delta_1 \leq C, \label{Q_probbound0} \\
&1 -  2\sa\ua\sum_{m=n_0+1}^ne^{-D\delta_1/\beta_{n_0}(m)} , \ \delta_1 > C, \label{Q_probbound}
\end{align}
where $C=e^{\left(2\left(1+\|Q_{N}\|_\infty+\frac{1}{1-\alpha}\right)+c_2\right)}$ and $\alpha=1-(1-\gamma)\pi_{min}$.
\end{theorem}

\begin{proof}[Proof Sketch] The proof for Theorem \ref{thm:conc} follows similarly to the proof of Corollary 2 from \cite{Chandak}. For this, we first rewrite \eqref{Qlearn2} in the form of  \citep[eqn.\ (40)]{Chandak}:
\begin{equation*}
    Q_{n+1}=Q_n+a(n)(G(Q_n,S_n,U_n)-Q_n+M_{n+1}+\zeta(Q_n,O^n,U^n)),
\end{equation*}
where $G(Q_n,S_n,U_n)=F(Q_n,S_n,U_n)+Q_n$ and $\zeta(Q_n,O^n,U^n)$ is an additional term due to non-Markovianity. Note that $G$ satisfies
$$\|\sum_{z\in\ZA}\tilde{\pi}(z)(G(q,z)-G(q',z))\|_\infty\leq \alpha \|z-z'\|_\infty,$$ 
where $\alpha=1-(1-\gamma)\pi_{min}$ (see \cite{Chandak}, eqn.\ (41)). This implies that $\sum_z\tilde{\pi}(z)G(\cdot,z)$ is a contractive mapping with contraction factor $\alpha$. It can be easily shown that its unique fixed point is $Q^*$. 

Now define $R_n$ for $n\geq n_0$ by $$R_{n+1}=R_n+a(n)\left(\sum_{z\in\ZA}\tilde{\pi}(z)G(R_n,z)-R_n\right),$$ 
where $R_{n_0}=Q_{n_0}$. Note that 
$$\|Q_n-Q^*\|_\infty\leq \|Q_n-R_n\|_\infty+\|R_n-Q^*\|_\infty.$$
The term $\|R_n-Q^*\|$ is bounded exactly as in \cite{Chandak}, and hence is not repeated here. The difference comes in the term $\|Q_n-R_n\|_\infty$, which is simplified as follows.
\begin{align}
Q_{n+1}-R_{n+1}&=(1-a(n))(Q_n-R_n)+a(n)(\zeta(Q_n,Y_n)+M_{n+1})\nonumber\\
&\;\;\;+a(n)\left(G(Q_n,Z_n)-\sum_{z\in\ZA}\tilde{\pi}(z)G(R_n,z)\right)\nonumber\\
&=(1-a(n))(Q_n-R_n)+a(n)(\zeta(Q_n,Y_n)+M_{n+1})\nonumber\\
&\;\;\;+a(n)\left(\sum_{z\in\ZA}\tilde{\pi}(z)(G(Q_n,z)-G(R_n,z))\right)\nonumber\\
&\;\;\;+ a(n)\left(G(Q_n,Z_n)-\sum_{z\in\ZA}\tilde{\pi}(z)G(Q_n,z)\right). \label{simpli-1}
\end{align}

To simplify the last term in the above equation, we use $V$, the solution of the Poisson equation defined above. Note that 
\begin{align}
&G(Q_n,Z_n)-\sum_{z\in\ZA}\tilde{\pi}(z)G(Q_n,z)\nonumber\\
&=F(Q_n,Z_n)-\sum_{z\in\ZA}\tilde{\pi}(z)F(Q_n,z)\nonumber\\
&=V(Q_n,Z_n)-\sum_{z'\in\ZA}\psi(z'\mid Z_n)V(Q_n,z')\nonumber\\
&= V(Q_n,Z_n) - E[V(Q_n,Z_{n+1})\mid Y_n]\nonumber\\
&\;\;\;\;\;\;\;\;\;\;+E[V(Q_n,Z_{n+1})\mid Y_n]-E[V(Q_n,Z_{n+1})\mid Z_n]. \label{simpli-2}
\end{align}

Here the second equality follows from the Poisson equation. For some $n\geq n_0$, iterating (\ref{simpli-1}) for $n_0\leq m\leq n$ and substituting (\ref{simpli-2}) gives us
\begin{align}
&Q_{n+1}-R_{n+1}=\sum_{k=n_0}^n \chi(n,k+1)a(k)M_{k+1} \nonumber\\
&\;\;\;\;\;\;\;\;\;\;+ \sum_{k=n_0}^n\chi(n,k+1)a(k)\left(\sum_{z\in\ZA}\tilde{\pi}(z)(G(Q_k,z)-G(R_k,z))\right) \nonumber\\
&\;\;\;\;\;\;\;\;\;\;+ \sum_{k=n_0}^n\chi(n,k+1)a(k)\left(V(Q_k,Z_k) - E[V(Q_k,Z_{k+1})\mid Y_k]\right) \nonumber\\
&\;\;\;\;\;\;\;\;\;\;+ \sum_{k=n_0}^n\chi(n,k+1)a(k) \left(\zeta(Q_k,Y_k)+\omega(Q_k,Y_k)\right).
\end{align}
The first three terms are bounded exactly as in Theorem 1 of \cite{Chandak}, so we do not repeat it here. The last term is the total error due to non-Markovianity and is equal to $\Delta_{n_0}(n+1)$. This completes the proof for Theorem \ref{thm:conc}.
\end{proof}

\begin{remark}
The constants $c_1$ and $c_2$ depend on $\|Q_N\|$  which in turn has a bound depending on $\|Q_0\|$ that can be derived easily using the discrete Gronwall inequality under our assumptions. Also note that if $a(n)$ are decreasing and $a(0) < 1$, then we can take $N=0$. In that case, we directly have an all-time bound using $n_0=0$. Calculations in \cite{Chandak} show that $c_1$ depends on quantities that depend on the mean hitting time of a fixed state, which is related to the mixing time of the Markov chain formed by the transition probabilities $\psi(s',u'\mid s,u)$. As pointed out in \textit{ibid.,} Section 4, it is also possible to stitch our results with known finite time bounds to get an `all-time' bound.
\end{remark} 

{\color{black} In the above result, all but the term $\Delta_{n_0}(n)$ are exactly as in the Markov case, therefore this is the term that needs to be kept small by a judicious choice of the agent state and dynamics in (\ref{state-eq}). Observe also that both $\zeta_n(\cdot,\cdot)$ and $\omega_n(\cdot,\cdot)$ are bounded uniformly in $n$. Hence for a suitable constant $\check{K} > 0$,
$$\|\Delta_{n_0}(n)\|_\infty \leq  \check{K}\sum_{m=n_0}^{n-1}e^{-\sum_{k=m+1}^na(k)}a(m) <\infty.$$
Also, it is a zero mean stationary sequence.  Then under mild conditions, it can be shown to have no effect on the a.s.\ convergence of the algorithm (\cite{Kush}, Chapter 6). Unfortunately, a good concentration bound for this class of `noise' seems unavailable and unlikely without some additional mixing conditions. }

Our objective is  to use these observations to motivate a particular philosophy for agent design. In the same spirit, the stationarity assumption (which could be relaxed to asymptotic stationarity) is also a part of this motivational exercise. We introduce this performance criterion and the ensuing notion of \textit{Recursively Computable Approximate Sufficient Statistics} (RCASS) in the next section.

\section{Recursively Computable  Approximate Sufficient Statistics}

In view of Theorem \ref{thm:approx}, Theorem \ref{thm:conc} and the discussion that follows, a legitimate criterion for agent design is that we need the conditional distribution of $(S_{n+1}, O_{n+1})$ given $(O^n, U^n)$ to be well approximated by its conditional distribution given $(S_n, U_n)$ alone. In other words, we need $S_n$ to be an `approximate sufficient statistics'. We shall call a sequence $\{Y_n\}$ taking values in a finite set $\SA_1$ to be a \textit{Recursively Computable Approximate Sufficient Statistics} (RCASS) if it is of the form $Y_n = g_1(\Gamma_n)$ for some function $g_1 : \SA_2 \mapsto \SA_1$ for some finite set $\SA_2$, where $\{\Gamma_n\} \subset \SA_2$ is given by a recursion of the  form
\begin{equation}
\Gamma_{n+1} = h_1(\Gamma_n, U_n, O_{n+1}), \ n \geq 0, \label{Rcass}
\end{equation}
for some $h_1: \SA_2\times\U\times\OA\mapsto \SA_2$. Clearly, $\{S_n\}$ above is a RCASS where $\SA_1 = \SA_2 = \SA$ and $g_1$ is the identity map. That $\SA_1, \SA_2$ be finite is our choice. In any case, note that $Y_n$ is a function of $(O^n,U^n)$ and therefore any conditional probability $\tilde{p}(O_{n+1} = o \mid  Y_n, U_n)$ can also be written alternatively as $\breve{p}(O_{n+1} = o \mid  O^n,U^n)$ for a suitable $\breve{p}( \cdot \mid  \cdot )$. With this backdrop, we have the following result.\\

\begin{lemma}\label{four} The transition probability function $p( \cdot  \mid  \cdot ) : \OA^\infty\times\U^\infty \mapsto \PA(\OA)$ can be uniformly approximated by a suitable choice of $\SA_1$ and  $\tilde{p}( \cdot  \mid \cdot ) : \SA_1\times\U \mapsto \PA(\OA)$ in the sense that the error $\|p( \cdot \mid  \cdot ) - \breve{p}( \cdot  \mid  \cdot )\|$ can be made smaller than any prescribed $\epsilon > 0$ by a suitable choice of $\tilde{p}$, equivalently, of $\breve{p}$, which is related to $\tilde{p}$ as described above.
\end{lemma}

\begin{proof} If we drop the requirement that $\breve{p}(i \mid  \cdot)$ be a probability and require only that it be a continuous function, then the result is immediate from the Stone-Weierstrass theorem for $S_1 = S_2  := \OA^N$ and $Y_n = \Gamma_n := [(O_n,U_{n-1}), \cdots, (O_{n-N},U_{n-N-1})]$ with $N$ large enough. The resulting vector can then be projected to $\PA(\SA)\times(\SA\times\U)^N$ to get the desired claim.
\end{proof}

This justifies restricting our attention to RCASS. Since $\sigma(Y_m, m \leq n) \subset \sigma(Z_m, m \leq n) \ \forall n$, we might as well take $g_1$ to be the identity map as above. This justifies our model for agent dynamics given by (\ref{state-eq}). 
It is instructive here to draw a parallel with the theory of POMDPs in classical stochastic control \citep{Kumar}. To make the parallels apparent, we shall repeat the same notation. In a POMDP, one controls a (controlled) Markov chain $\{X_n\}$  on a finite state space $\SA_3$ given a process of observations $\{O_n\}$ taking values in another finite space, with a controlled transition kernel (say) $p'(s',o\mid s,u), s, s' \in \SA,$ that satisfies $\forall n$,
\begin{equation}
P(X_{n+1} = s', O_{n+1} = o \mid  X^n,U^n,O^n) = p'(s', o\mid  S_n,U_n). \label{POMDP}
\end{equation}
The control $U_n$ is allowed to depend only upon $O^n, U^{n-1}$, and independent extraneous randomization. Letting $\pi_n(\cdot)$ denote the conditional law of $S_n$ given $(O^n,U^n)$, the process $\{\pi_n\}$ of belief states is given recursively by the nonlinear filter
\begin{equation}
\pi_{n+1}(i) = \frac{\sum_j\pi_n(j)p'(i, O_{n+1}\mid  j, U_n)}{\sum_{j,k}\pi_n(j)p'(k, O_{n+1}\mid  j, U_n)}, \ \forall n. \label{NLF}
\end{equation}
This is an easy consequence of the Bayes rule. Furthermore, the process $\{\pi_n\}$ is a $\PA(\SA)$-valued controlled Markov chain. The equation (\ref{NLF}) has the same form as (\ref{state-eq}), except that $\pi_n$ resides in the probability simplex $\PA(\SA)$, which is not a finite set. Thus replacing it with (\ref{state-eq}) amounts to a finite state approximation to (\ref{NLF}). A further and critical  difference in our case is that (\ref{POMDP}) does not hold in our case. Therefore $\{\pi_n\}$ is \textit{not} a controlled Markov chain. This is the crux of the \textit{curse of non-Markovianity}.

{\color{black} Note that $\{O_n\}$ is trivially a POMDP if we view $X_n := (O^n,U^{n-1})$ as a controlled Markov chain in state space $\SA_2 := \OA^\infty\times\U^\infty$, with control process $\{U_n\}$. The transition at time $n$ given the control choice $U_n=u$ causes the chain to move to the next state $(O^{n+1} = (O^n,O_{n+1}=o), U^n = (U^{n-1},U_n=u))$ with 
\begin{eqnarray*}
\lefteqn{P\left(O^{n+1} = (O^n,u), U^n = (U^{n-1},u)| (O^n,U^{n-1}), U_n = u\right)} \\
 &=& p\left( \ o \ |O^n,(U^{n-1},u)\right)
\end{eqnarray*}
for $p( \cdot \ | \cdot, \cdot)$ as in \eqref{kernel}. The associated observation process $O_n = g'(X_n)$ where 
	$$g' : ((o_n, u_{n-1}), (o_{n-1}, u_{n-2}), \cdots) \in (\OA\times\U)^\infty \mapsto o_n \in \OA.$$}
The agent design task is to develop an approximate \textit{finite state} model of the original system as above, which in this case is the controlled nonlinear filter in view of the `separated control' formalism. This is reminiscent of the \textit{internal model principle} in adaptive control \citep{Francis} (see \cite{Bin} for a recent survey). An informal statement of this principle says: `a good controller incorporates a model of the dynamics that generate the signals which the control system is intended to track'. It should also be mentioned that a natural choice for this purpose that comes to mind is approximating $(O^n, U^{n-1})$ by a finite window $[O_{n-K}, U_{n-K}, \cdots , O_{n-1}, U_{n-1}, O_n]$. (See the proof of Lemma \ref{four} above.) This can be viewed as a special case of the framework described above. This situation has been extensively analyzed in \cite{Kara} and a Q-learning algorithm is proposed and analyzed in this framework therein. Another important work in this spirit is \cite{Aditya}, which postulates an \textit{approximate information state} that corresponds to our RCASS, by stating a priori desirable properties thereof that facilitate approximate dynamic programming. In contrast, we are motivated by the error decomposition in the concentration bound and take a more statistical viewpoint of seeking approximate sufficient statistics to mimic certain conditional laws.
Whereas we justified the need to well approximate the conditional laws using agent state in terms of the transients of the Q-learning algorithm, we can also motivate it by looking at the asymptotic behavior as characterized in (\ref{Singh1}).

{\color{black} Our input to the learning scheme is a stationary process controlled through a randomized control policy that chooses at time $n$ a control $U_n$ in $\U$ according to a probability distribution on $\U$  that depends only on the current agent state $S_n$. As argued earlier, the agent state $S_n$ will be a fixed function of the state $(O^n,U^{n-1})$ of the underlying $\OA^\infty\times\U^\infty$-valued controlled Markov chain, therefore so will be the conditional law of $U_n$ given  $(O^n,U^{n-1})$. Hence this is a stationary randomized policy for the underlying controlled Markov chain $(O^n,U^{n-1}), n > -\infty$. This, however, decides only the sampling distribution of state-action pairs for the Q-learning algorithm. (Note that the learning scheme is `off-line'.)  As shown in Theorem 1- Proposition 2, the `optimal choice' of the algorithm obtained from the Q-values learnt, will be a fixed function of $S_n$ at time $n$ and again by the above argument, it corresponds to a stationary policy. Since the state and action spaces of this $\OA^\infty\times\U^\infty$-valued chain are compact and the transition mechanism is continuous in state and action variables, standard dynamic programming arguments guarantee the sufficiency of stationary policies. But since we are working with an approximate sufficient statistics, the best one can hope for is a near-optimal stationary or stationary randomized policy. Since the agent state and its dynamics are only approximations of the nonlinear filter, one cannot in general expect this class to contain an optimal policy for the original problem, but in view of the interpretation given in Proposition 2, it will be near-optimal if the state aggregation indicated in Proposition 2 is fine enough. Thus there will be a trade-off between the complexity of the agent dynamics which we design, and the optimality gap.
}

{\color{black}\section{Autoencoder-based Computation Scheme}
We now propose a specific architecture inspired by the autoencoder model from \cite{Champion}. However, unlike their model, where one tries to match the output with the input, we match the output with the \textit{next} input, i.e., our error is not to be viewed as a decoding error, but as a prediction error. Our architecture has a `state' neural network that takes $S_{n-1},U_{n-1},O_{n}$ as input and gives the next agent state $S_n$ as output. This feeds into the `observation' neural network that takes this $S_n$ along with $U_n$ as input and outputs a prediction $\widehat{O}_{n+1}$ of the next observation $O_{n+1}$ as the output (see Figure \ref{fig:arch}). The neural networks are trained by comparing $O_{n+1}, \widehat{O}_{n+1}$. The criterion for this comparison can be chosen depending on the problem specifics. For example, it could be one of the classical loss functions penalizing the difference between  $O_{n+1}$ and $\widehat{O}_{n+1}$, or it may compare their distributions or conditional distributions (see, e.g., \cite{Osband}). See section 3.1 of \cite{Aditya} for some popular metrics for comparing probability measures. There are also other architectures that have been proposed with alternative underlying philosophies, e.g., \cite{Igl}. 

\begin{figure}[ht!]
    \centering
    \includegraphics[width=0.8\linewidth]{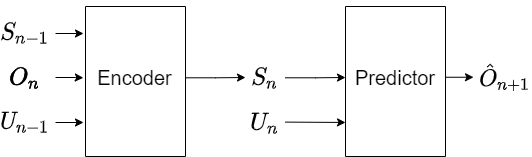}
    \caption{Learning the agent state dynamics using autoencoders}
    \label{fig:arch}
\end{figure}


\subsection{Numerical Experiments}
For numerical experiments, we use the above architecture with the squared $\ell_2$ distance between $O_{n+1}$ and $\widehat{O}_{n+1}$ as the loss. The actions $U_n$ are determined using a Deep Q-Network (DQN) \cite{dqn} which takes state $S_n$ as the input. We use experience replay and target network for efficient training of DQN \cite{dqn}. For training the combined system of autoencoder and DQN, we take an alternating training approach, i.e., we first fix the DQN and train the autoencoder for a few epochs and then train the DQN for a few epochs while keeping the autoencoder fixed. This alternating training is then repeated. This is necessary as DQN takes the output of autoencoder as input and vice versa, and hence training them together in a single epoch creates additional fluctuations. We call our agent (trained in the above described manner) as the Non-Markovian Q-Agent (NMQ Agent).  

We test our agent on partially-observed versions of two standard environments - Cartpole \cite{popgym} and Mountain Car \cite{sutton} and on a controlled Non-Markovian random walk. 

\begin{figure*}[!h]
\centering
\begin{subfigure}{.4\textwidth}
  \centering
  \includegraphics[width=0.99\linewidth]{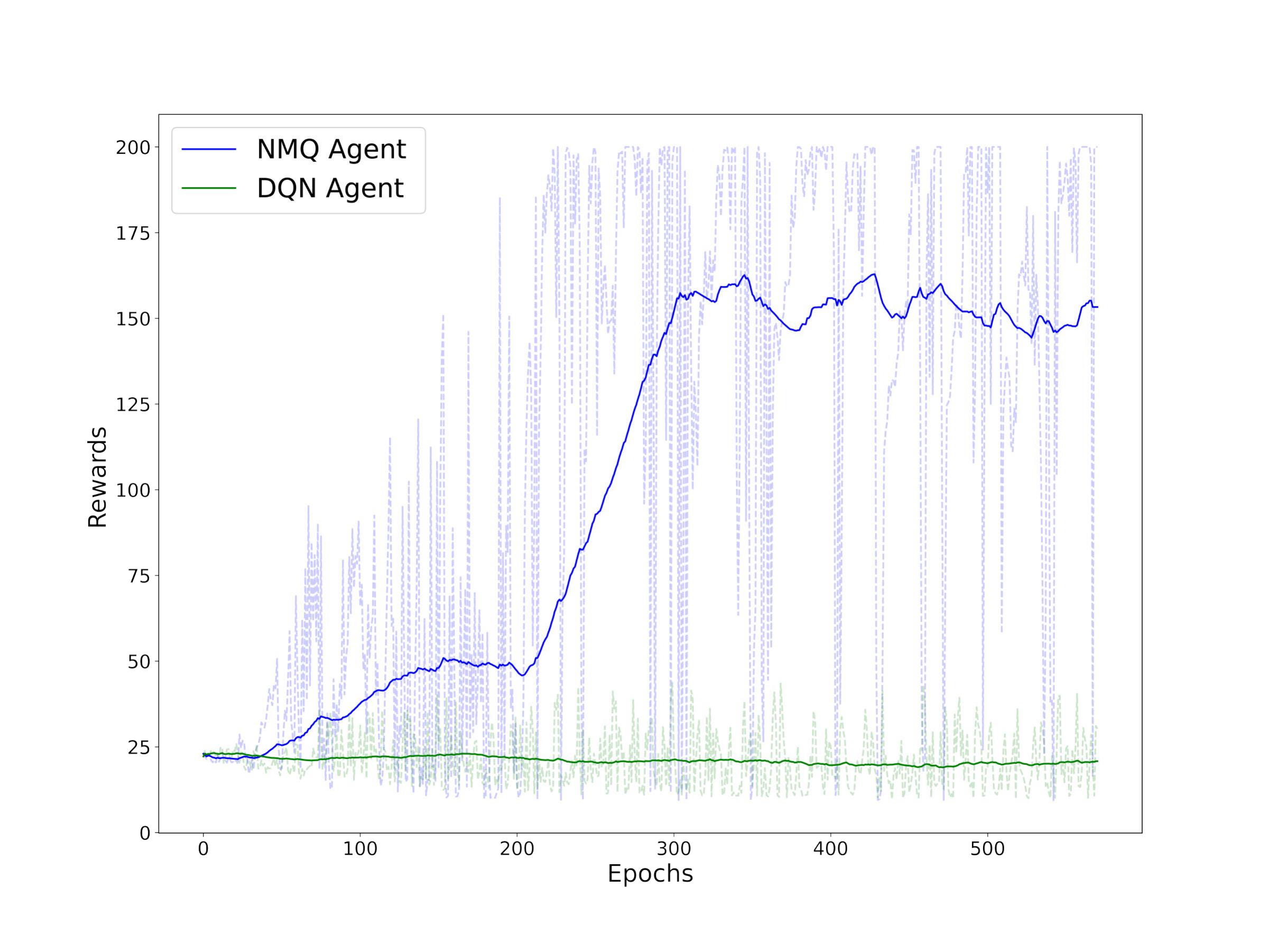}
  \caption{Cartpole}
  \label{fig:cartpole}
\end{subfigure}%
\begin{subfigure}{.4\textwidth}
  \centering
  \includegraphics[width=.99\linewidth]{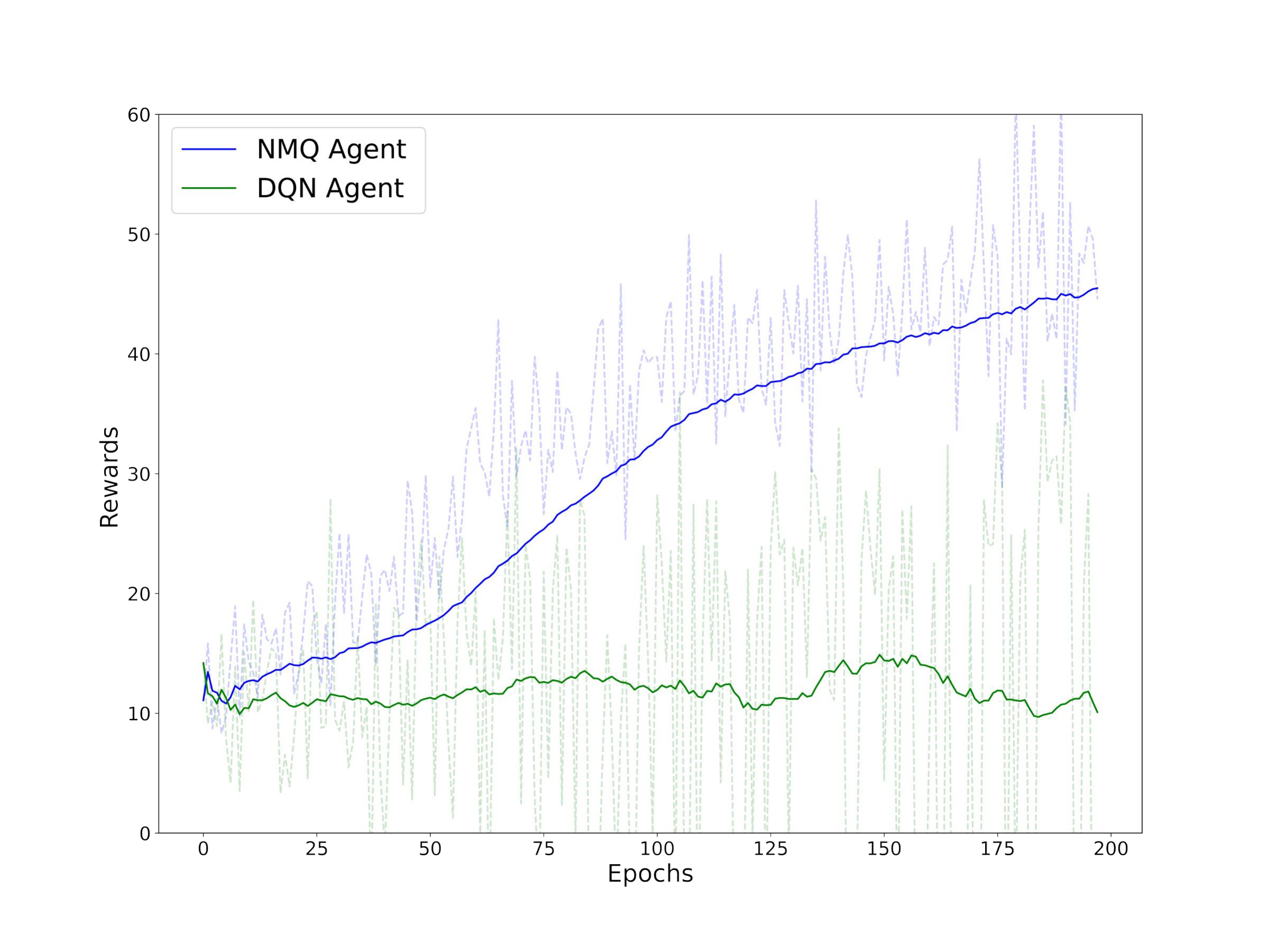}
  \caption{Mountain Car}
  \label{fig:mountain}
\end{subfigure}
\begin{subfigure}{.4\textwidth}
  \centering
  \includegraphics[width=.99\linewidth]{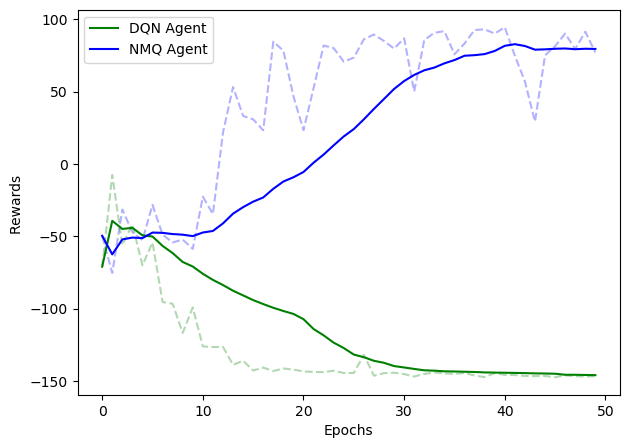}
  \caption{Non-Markovian Random Walk}
  \label{fig:walk}
\end{subfigure}
\caption{Moving average of episodic reward for a single run of Non-Markovian Q-agent and the DQN-agent on the environments (a) cartpole, (b) mountain car and (c) Non-Markovian random walk. The dotted lines represent the reward obtained in each episode.}
\label{fig:result}
\end{figure*} 

\begin{itemize}
    \item \textbf{Cartpole:} Cartpole is the simulation of a cart with an upside-down pole where the agent has to maintain the vertical angle of the pole sufficiently small while remaining in a restricted horizontal displacement region. The system state consists of $4$ dimensions: horizontal displacement and velocity, and angular displacement and velocity. The agent is just able to observe the displacements (horizontal and angular). The agent's action space consists of two actions: applying force on left or right. The episode runs for a maximum of 200 timesteps or stops when either of the horizontal or angular displacement constraints are violated and the agent receives reward $1$ at each timestep until the episode ends. The environment is said to be solved if the agent achieves a reward of $200$ (the maximum reward).
    \item \textbf{Mountain Car:} This environment simulates an under-powered car which must drive up a steep hill. The car is situated in a valley and must learn to drive up the opposite hill before it can make it to the top of the desired hill. The state space consists of velocity and position, and in our partially observed version, the agent can just observe the position. The action space consists of 3 actions: force towards left or right or no force. We use a common variation for the reward function where the reward depends on the distance of the car from the hill, increasing as the car reaches close to the top of the hill, and also on the absolute velocity of the car.
    {\color{black} \item \textbf{Non-Markovian Random Walk:} This environment simulates a controlled random walk where the next position depends not only on the current position and action but also on the past positions. Here the position at each time $n$ is the observation $O_n$ and the actions $U_n$ can be in the set $\{-5,-4,\ldots,4,5\}$. Let $sum_n$ denote the sum of past positions, i.e., $sum_n=\sum_{m=0}^n O_m$. Then the next observation is given by
    $O_{n+1}=O_n+(sum_n\%10)-5+U_n,$
    where $\%$ denotes the modulus operator. The agent is initialized at a random position between $-100$ and $100$ and receives high positive reward of 100 on reaching $0$, and slight negative reward of $-|O_n|/100$ otherwise. The episode ends either when the agent reaches $0$ or after 150 timesteps, if the agent fails to reach $0$ before then.}
\end{itemize}
Figure \ref{fig:result} compares the performance of our NMQ agent with a DQN agent where the observations are directly given as inputs to the Deep Q-Network. For all experiments, the plots are for a single run of the environment. We plot the moving average of the episodic rewards and the light dotted line represents the reward in each episode. Results show that the NMQ agent is able to `learn' and solve the environments, while it seems that the DQN fails to learn anything significant. These are plots for a single run and the fluctuations in episodic rewards are due to the partially-observed setting \cite{popgym}.

}
\section{Future Directions}

\subsection{The assumption of stationarity}

One of the limitations of the above analysis has been the assumption of stationarity. The key consequence of it for us was to justify the time-homogeneous transition probabilities $q( \cdot \mid  \cdot, \cdot)$. As already observed, asymptotic stationarity suffices.  In the  absence of this assumption, the aforementioned time-homogeneity itself will be an \textit{assumption}. The justification for using the Poisson equation as stated in our analysis is also no longer valid without this additional assumption. In general, analyzing reinforcement learning in non-stationary environments is a major open problem, except perhaps when the environment is slowly varying and therefore quasi-stationary, whence it becomes a tracking issue. See \cite{Yash} for some recent work in this direction and \cite{Cao} for analysis of nonstationary stochastic control problems. 

\subsection{RKHS-based Computation scheme}
 Another method for approximating agent dynamics given by $P(S_n\mid O^n)$, is to use Hilbert space embeddings of conditional distributions \citep{Song_RKHS_main, Song_RKHS_survey, RKHS_review}. This involves representing probability distributions by elements in a reproducing kernel Hilbert space (RKHS). First, we give a brief overview of how distributions can be embedded into an RKHS. An RKHS $\F$ with kernel $k$ on $\mathcal{X}$ is a Hilbert space of functions $f:\mathcal{X}\mapsto \mathbb{R}$ with inner product $\langle\cdot,\cdot\rangle_\F$ that satisfies the property $\langle f(\cdot),k(x,\cdot)\rangle_\F=f(x)$. Examples of such kernel functions include polynomial kernels and Gaussian kernels. These kernels are associated with feature maps $\phi : \X \mapsto \F$ such that $k(x,x')=\langle \phi(x),\phi(x')\rangle_\F$. Once we have a fixed feature map, for distribution $P(X)$ over $\mathcal{X}$, we can define the embedding ``mean map" $\mu_X\coloneqq E_X[\phi(X)]$. Similarly, we can define the empirical mean map $\hat{\mu}_X $ which converges with high probability to $\mu_X$ in the RKHS norm at a rate of $\mathcal{O}_p(1/\sqrt{m})$, where $m$ denotes the number of samples. This mean map $\mu_X$ can be used to represent a distribution and under certain assumptions on the kernel $k$, we can guarantee that distinct distributions map to distinct mean maps in RKHS. 

Now, to represent conditional distributions as embeddings, we use the cross-covariance operator $C_{XY}\coloneqq E_{XY}[\phi(X)\otimes\varphi(Y)]-\mu_X\otimes\mu_Y,$
where $\otimes$ denotes the tensor product and $\phi(\cdot)$ and $\varphi(\cdot)$ are feature maps on $\mathcal{X}$ and $\mathcal{Y}$, respectively. Then the conditional map operator $\U_{Y\mid X}$ is defined as $\U_{Y\mid X}=C_{YX}C_{XX}^{-1}$ and the conditional mapping given $X=x$ is given by $\mu_{Y\mid x}=\U_{Y\mid X}\phi(x)$. The cross covariance operators can be empirically estimated recursively.

We can model our problem of approximating $P(S_n\mid H_n)$ as a hidden Markov model, where $\{S_n\}$ are modeled as states from a Markov chain (not observed) and the observations $O_n$ satisfy $P(O_n\in\mathcal{A}\mid O^{n-1},S^n)=P(O_n\in\mathcal{A}\mid S_n)$. Under this model, the conditional distribution of $(S_{n+1},O_{n+1})$ given $(O^n,U^n)$ would be well approximated by its conditional distribution given $(S_n, U_n)$. \cite{Song_RKHS_main} gives an algorithm for approximating $\mu_{S_{n+1}\mid o^{n+1}}$ in a recursive manner, and then inferring $s_{n+1}$ from this conditional distribution. They first show that $$\mu_{S_{n+1}\mid s_n,o_{n+1}}=\T_1\varphi(s_n)+\T_2\phi(o_{n+1}),$$ where $\T_1$ and $\T_2$ are conditional map operators. Then our required distribution can be approximated as follows $$\mu_{S_{n+1}\mid o^{n+1}}\approx \T_1\mu_{S_n\mid o^{n}}+\T_2\phi(o_{n+1}).$$ This decomposes the update into two simple operations: the first term here propagates from time $n$ to $n+1$ and the second term accounts for the current observation. Estimates of these operators ($\hat{\T}_1$ and $\hat{\T}_2$) can be learned empirically \citep[Alg.\ 1]{Song_RKHS_main} and can then be used for inference \citep[Alg.\ 2]{Song_RKHS_main}.

\section*{Acknowledgement}
VSB was supported by the S.\ S.\ Bhatnagar Fellowship from the Council of Scientific and Industrial Research, Government of India. An early version of this work was presented at the workshop on `Learning-based Control of Queues and Networks', June 6, 2022, associated with SIGMETRICS 2022, Mumbai, and a short summary was presented in the INFORMS Applied Probability Conference, Nancy, France, on June 29th, 2023.

\appendix
\section{A Bound on $\Delta(n)$}\label{appendix-bound}
Motivated by \cite{Paulin}, we give a bound on $\Delta(n)$. We need certain assumptions on the observation process and the state mapping for the bound to hold. We first define a matrix that measures the strength of dependence between the random variables. Define $W_1=(O^1,U^1)$ and $W_m=(O_m,U_m)$ for all $m>1$. As defined before, $Z_m=(S_m,U_m)$ and $Z_m=g(W_1,\ldots, W_m)$. Fix $n\geq 1$. Let $p^{w_1,\ldots,w_i}(\cdot)$ denote the conditional probability $P(\cdot\mid W_1=w_1,\ldots,W_i=w_i)$. Then we have the following four assumptions:
\begin{enumerate}
    \item Define the upper triangular matrix $\Phi$ where $\Phi_{i,i}=1$ for $i\leq n$, $\Phi_{j,i}=0$ for $1\leq i<j\leq n$ and, for $1\leq i<j\leq n,$
\begin{align*}
    \Phi_{i,j}=\sup_{w_1,\ldots,w_i,w_i'}d_{TV}\Big(p^{w_1,\ldots,w_{i-1},w_i}&(Z_j,O_{j+1}),\\ &p^{w_1,\ldots,w_{i-1},w'_i}(Z_j,O_{j+1})\Big),
\end{align*}
where $d_{TV}(\cdot,\cdot)$ denotes the total variation distance between probability distributions. We assume that $\|\Phi\|_2\leq d_4$ where $d_4>0$. Here $\|\Phi\|_2$ denotes the matrix norm induced by the euclidean norm. 
    \item Similarly, define the upper triangular matrix $\Psi$ where $\Psi_{j,i}=0$ for $1\leq i<j\leq n$ and 
    \begin{align*}
    \Psi_{i,j}=\sup_{w_1,\ldots,w_i,w_i'}d_{TV}\Big(p^{w_1,\ldots,w_{i-1},w_i}&(Z_j,Z_{j+1}),\\ &p^{w_1,\ldots,w_{i-1},w'_i}(Z_j,Z_{j+1})\Big),
    \end{align*}
     for $1\leq i\leq j\leq n$. We again assume that $\|\Psi\|_2\leq d_4$.
\end{enumerate}
While these exact assumptions have been chosen to suit the proof of Lemma \ref{lemma-delta-bound}, the intuition behind them is that these assumptions restrict the dependence in the observation process and between observations and mapped states. Before stating the lemma, we give a hidden Markov model-based example adapted from \citep[Example 2.15]{Paulin}, which justifies our assumptions. Let $\{\tilde{O}_j\}$ be a Markov chain and let $\{O_j\}$ be random variables such that $$P(O_n\in\mathcal{A}\mid O_1,\ldots, O_{n-1}, \tilde{O}_1,\ldots, \tilde{O}_n)=P(O_n\in\mathcal{A}\mid  \tilde{O}_n).$$ Then the value $\Phi_{i,j}$ can be bounded by the sum of 
    $$\Phi'_{i,j}\coloneqq\sup_{w_1,\ldots,w_i,w_i'}d_{TV}\left(p^{w_1,\ldots,w_{i-1},w_i}(O_{j+1}), p^{w_1,\ldots,w_{i-1},w'_i}(O_{j+1})\right),$$
    and 
    $$\Phi''_{i,j}\coloneqq\sup_{w_1,\ldots,w_i,w_i', o_{j+1}}d_{TV}\left(p^{w_1,\ldots,w_{i-1},w_i, o_{j+1}}(Z_j), p^{w_1,\ldots,w_{i-1},w'_i, o_{j+1}}(Z_j)\right).$$
    Then the norm of the matrix $\Phi'$ formed by the first term above is bounded by a constant proportional to the mixing time of the underlying Markov chain \citep[Corollary 2.16]{Paulin}. The norm of the matrix $\Phi''$ can be bounded by making suitable assumptions on the state mapping. $\|\Psi\|$ can be bounded similarly.



\begin{remark}
We have made the assumption that $\|\Phi\|_2\leq d_4$ and $\|\Psi\|_2\leq d_4$ to improve readability. They can be generalized to $\|\Phi\|_2\leq d_4n^{d_5}$ (and similarly for $\Psi$), where $2d_5<2d_2-1$. This would give $n^{1-2d_2+2d_5}$ (instead of $n^{1-2d_2}$) in the denominator of the bound in Lemma \ref{lemma-delta-bound}. 
\end{remark}

\begin{lemma}\label{lemma-delta-bound}
There exists constant $c_7>0$ such that for $n\geq 1$,
$$P(\|\Delta(n)\|_\infty\geq \delta)\leq 2\sa\ua\exp\left(\frac{-c_7\delta^2}{n^{1-2d_2}}\right).$$
\end{lemma}
\begin{proof}
For a fixed $s,u$, we know that 
\begin{align*}
&\Delta^{s,u}(n)=\sum_{m=1}^{n-1}\chi(n-1,m+1)a(m)\times\\
&\Bigg(I\{S_m=s,U_m=u\}\bigg(E[r(S_m,U_m,O_{m+1})\mid O^m,U^m]\\
&\;\;\;\;\;\;\;\;\;\;\;\;\;\;-E[r(S_m,U_m,O_{m+1})\mid S_m,U_m]\bigg)\\
&+\gamma I\{S_m=s,U_m=u\}\bigg(E[\max_aQ_m(S_{m+1},a)\mid O^m,U^m]\\
&\;\;\;\;\;\;\;\;\;\;\;\;\;\;-E[\max_aQ_m(S_{m+1},a)\mid S_m,U_m]\bigg)\\
&+\bigg(E[V^{s,u}(Q_m,Z_{m+1})\mid O^m,U^m]-E[V^{s,u}(Q_m,Z_{m+1})\mid S_m,U_m]\bigg)\Bigg).
\end{align*}

For ease of notation, we omit $s,u$ from the left-hand side and define
\begin{eqnarray*}
\Bb(W_1,\ldots,W_n)=\sum_{m=1}^{n-1}\chi(n-1,m+1)a(m)B_m(W_1,\ldots,W_m),
\end{eqnarray*}
and
\begin{eqnarray*}
\Cb(W_1,\ldots,W_n)=\sum_{m=1}^{n-1}\chi(n-1,m+1)a(m)C_m(S_m,U_m),
\end{eqnarray*}
where 
\begin{align*}
&B_m(W_1,\ldots,W_m)\\
&=I\{S_m=s,U_m=u\}E[r(S_m,U_m,O_{m+1})\mid O^m,U^m]\\
&\;\;\;+\gamma I\{S_m=s,U_m=u\}E[\max_aQ_m(S_{m+1},a)\mid O^m,U^m]\\
&\;\;\;+E[V^{s,u}(Q_m,Z_{m+1})\mid O^m,U^m]
\end{align*}
and 
\begin{align*}
C_m(S_m,U_m)&=I\{S_m=s,U_m=u\}E[r(S_m,U_m,O_{m+1})\mid S_m,U_m]\\
+&\gamma I\{S_m=s,U_m=u\}E[\max_aQ_m(S_{m+1},a)\mid S_m,U_m]\\
+&E[V^{s,u}(Q_m,Z_{m+1})\mid S_m,U_m].
\end{align*}

We first give a bound on $P(|\Cb(W)-E[\Cb(W)]|\geq\delta)$. Note that $C_m(S_m,U_m)$ is bounded by $C_{max}\coloneqq(\frac{1}{1-\gamma}+V_{max})$ for all $m$, where $V_{max}$ is an upper bound on the term $V^{s,u}(\cdot,\cdot)$ for all $s,u$. One such bound is derived in \cite{Chandak}. Define $\F_i=\sigma(W_1,\ldots,W_i)$ for $i\leq n$, and write $\Cb(W)-E[\Cb(W)]=\sum_{i=1}^n\Cf_i(W)$ with $\Cf_i(W)=E[\Cb(W)\mid \F_i]-E[\Cb(W)\mid \F_{i-1}]$. Then
\begin{align*}
    \Cf_i(W)&\leq \sup_aE[\Cb(W)\mid W_1,\ldots,W_{i-1},W_i=a]\\
    &\;\;\;\;\;\;-\inf_bE[\Cb(W)\mid W_1,\ldots,W_{i-1},W_i=b].
\end{align*}
Let $\Cf_{i,max}(W)$ and $\Cf_{i,min}(W)$ denote the supremum and the infimum terms on the right, respectively. Also, assume that these values are achieved at $a$ and $b$ respectively. Then
\begin{align*}
&\Cf_{i,max}(W)-\Cf_{i,min}(W)\\
&=\sum_{j=i}^{n-1}\chi(n-1,j+1)a(j) \Big(E[C_j(S_j,U_j)\mid W_1,\ldots,W_{i-1},W_i=a]\\
&\;\;\;\;\;\;\;\;\;\;\;\;\;\;\;\;\;\;\;-E[C_j(S_j,U_j)\mid W_1,\ldots,W_{i-1},W_i=b]\Big)\\
&\leq \sum_{j=i}^{n-1}\chi(n-1,j+1)a(j)C_{max}2\Psi_{i,j}\\
&\leq 2C_{max}\sum_{j=i}^{n-1}\Psi_{i,j}\kappa_j,
\end{align*}
where $\kappa_j=\chi(n-1,j+1)a(j)$. 
Using the Azuma-Hoeffding inequality, we obtain 
\begin{align*}
P(|\Cb(W)-E[\Cb(W)]|\geq \delta)&\leq \exp\left(\frac{-2\delta^2}{\sum_{i=1}^n\left(\sum_{j=i}^n2C_{max}\Psi_{i,j}\kappa_j\right)^2}\right)\\
&=\exp\left(\frac{-2\delta^2}{4C_{max}^2\|\Psi\kappa\|^2}\right)\\
&\leq\exp\left(\frac{-2\delta^2}{4C_{max}^2\|\Psi\|^2\|\kappa\|^2}\right).
\end{align*}
Using the assumption $\frac{d_1}{n}\leq a(n)\leq d_3\left(\frac{1}{n}\right)^{d_2},$ 
\begin{eqnarray*}
\chi(n-1,j+1)&=&\prod_{\ell=j+1}^{n-1} (1-a(\ell))\leq \exp\left(-\sum_{\ell=j+1}^{n-1} a(\ell)\right)\\
&\leq& \exp\left(-\sum_{\ell=j+1}^{n-1} \frac{d_1}{\ell}\right)\leq \left(\frac{j+1}{n-1}\right)^{d_1}.
\end{eqnarray*}
Then
$$\chi(n-1,j+1)a(j)\leq d_3\left(\frac{1}{j}\right)^{d_2}\left(\frac{j+1}{n-1}\right)^{d_1}\leq \frac{d_32^{d_1}j^{d_1-d_2}}{(n-1)^{d_1}},$$ and finally
$$\|\kappa\|_2^2\leq \frac{d_3^24^{d_1}}{(n-1)^{2d_1}}\sum_{j=1}^{n-1} j^{2d_1-2d_2} \leq c_4\left(\frac{1}{n}\right)^{2d_2-1},$$
for some constant $c_4$ (assuming $n>1$). This implies that
\begin{equation}\label{bound-C}
    P(|\Cb(W)-E[\Cb(W)]|\geq \delta)\leq \exp\left(\frac{-c_5\delta^2}{n^{1-2d_2}}\right),
\end{equation} for some constant $c_5$. 

We now give a bound on $P(|\Bb(W)-E[\Bb(W)]|\geq\delta)$. 
As before, $\Bb(W)-E[\Bb(W)]=\sum_{i=1}^n\Bf_i(W)$ where $$\Bf_i(W)=E[\Bb(W)\mid \F_i]-E[\Bb(W)\mid \F_{i-1}].$$ Then
\begin{align*}
    &\Bf_i(W)\\
    &=E[\Bb(W)\mid \F_i]-E[\Bb(W)\mid \F_{i-1}]\\
    &=\sum_{j=1}^{n-1}\chi(n-1,j+1)a(j) \bigg(E[B_j(W_1,\ldots,W_j)\mid \F_i]\\
    &\;\;\;\;\;\;\;\;\;\;\;\;\;\;\;\;\;\;\;\;\;\;\;\;-E[B_j(W_1,\ldots,W_j)\mid \F_{i-1}]\bigg)\\ 
    &=\sum_{j=i}^{n-1} \chi(n-1,j+1)a(j)\times\\
    &\;\;\;\;\Bigg(E\bigg[V^{s,u}(Q_j,Z_{j+1})+I\{S_j=s,U_j=u\}r(S_j,U_j,O_{j+1})\\
    &\;\;\;\;\;\;\;\;\;\;\;\;+\gamma I\{S_j=s,U_j=u\} \max_aQ_j(S_{j+1},a)\mathrel{\bigg|} \F_i\bigg]\\
    &\;\;-E\bigg[V^{s,u}(Q_j,Z_{j+1})+I\{S_j=s,U_j=u\}r(S_j,U_j,O_{j+1})\\
    &\;\;\;\;\;\;\;\;\;\;\;\;+\gamma I\{S_j=s,U_j=u\}\max_aQ_j(S_{j+1},a)\mathrel{\bigg|} \F_{i-1}\bigg]\Bigg).\\
\end{align*}
Here the last equality follows from the law of total expectation. Define 
\begin{align*}
&\Bf_{i,max}(W)=\sup_a \sum_{j=i}^{n-1} \chi(n-1,j+1)a(j)\times\\
&E\Bigg[V^{s,u}(Q_j,Z_{j+1})+I\{S_j=s,U_j=u\}r(S_j,U_j,O_{j+1})\\
&+I\{S_j=s,U_j=u\}\gamma \max_aQ_j(S_{j+1},a)\mid W_1,\ldots,W_{i-1},W_i=a\Bigg],
\end{align*}
 and 
\begin{align*}
&\Bf_{i,min}(W)=\inf_b \sum_{j=i}^{n-1} \chi(n-1,j+1)a(j)\times\\
&E\Bigg[V^{s,u}(Q_j,Z_{j+1})+I\{S_j=s,U_j=u\}r(S_j,U_j,O_{j+1})\\
&+I\{S_j=s,U_j=u\}\gamma \max_aQ_j(S_{j+1},a)\mid W_1,\ldots,W_{i-1},W_i=b\Bigg].
\end{align*}
Note that $r(S_j,U_j,O_{j+1})\leq 1$ and $\gamma Q_j(S_{j+1},a)\leq \frac{\gamma}{1-\gamma}$ for all $1\leq j\leq n-1$ and $a\in\mathcal{U}$. Then using the definition of $\Phi$ and $\Psi$, we have
\begin{align*}
&\Bf_i(W)\leq \Bf_{i,max}(W)-\Bf_{i,min}(W)\\
&\leq \sum_{j=i}^{n-1} \chi(n-1,j+1)a(j)\left(2\Phi_{i,j}\kappa_j+\left(\frac{\gamma}{1-\gamma}+V_{max}\right)2\Psi_{i,j}\kappa_j\right).
\end{align*}
Applying Azuma-Hoeffding inequality again, we have as before
\begin{equation}\label{bound-B}
    P(|\Bb(W)-E[\Bb(W)]|\geq\delta)\leq \exp\left(\frac{-c_6\delta^2}{n^{1-2d_2}}\right),
\end{equation}
where $c_6$ is some constant. Note that $E[\Bb(W)]=E[\Cb(W)]$. Then defining $c_7\coloneqq 1/4\min\{c_5,c_6\}$, applying union bound on (\ref{bound-C}) and (\ref{bound-B}), and summing over all $s\in\SA,u\in\U$, we get
$$P(\|\Delta(n)\|_\infty\geq \delta)\leq 2\sa\ua\exp\left(\frac{-c_7\delta^2}{n^{1-2d_2}}\right).$$
\end{proof}

\end{document}